\documentclass[twocolumn,showpacs,amsmath,amssymb]{revtex4}
\usepackage{graphicx}
\usepackage{slashed}
\usepackage{pifont}
\usepackage{color}
\usepackage{epsfig}
\usepackage{nicefrac}
\usepackage{amssymb,amsmath,amsthm}

\newcommand{\beqn}{\begin{equation}}
\newcommand{\eeqn}{\end{equation}}

\newtheorem{thm}[equation]{Theorem}

\begin{document}
\title{Transparency of Strong Gravitational Waves}
\author{Y. Hadad, V. Zakharov}
\affiliation{Department of Mathematics, University of Arizona, Tucson, Arizona, 85721 USA}


\date{January 28th, 2014}


\begin{abstract}
This paper studies diagonal spacetime metrics. It is shown that the overdetermined Einstein vacuum equations are compatible if one Killing vector exists. The stability of plane gravitational waves of the Robinson type is studied. This stability problem bares a fantastic mathematical resemblance to the stability of the Schwarzschild black hole studied by Regge and Wheeler. Just like for the Schwarzschild black hole, the Robinson gravitational waves are proven to be stable with respect to small perturbations. We conjecture that a bigger class of vacuum solutions are stable, among which are all gravitational solitons. Moreover, the stability analysis reveals a surprising fact: a wave barrier will be transparent to the Robinson waves, which therefore passes through the barrier freely. This is a hint of integrability of the $1+2$ vacuum Einstein equations for diagonal metrics.
\end{abstract}

\pacs{02.30.Ik,02.30.Jr,05.45.Yv,04.20.-q,04.30.-w,04.70.Bw}

\maketitle

\section{Introduction} \label{sec:Introduction}
In the theory of relativity, the Einstein-Hilbert action is
\begin{equation} \label{eq:EinsteinHilbertAction}
S = \frac{1}{2} \int R \sqrt{-g} d^4 x
\end{equation}
where $R$ is the scalar curvature of the spacetime metric $g_{\mu\nu}$, $g$ is the determinant of $g_{\mu\nu}$ and the integration is performed over the four-dimensional spacetime. Varying the Einstein-Hilbert action (\ref{eq:EinsteinHilbertAction}) with respect to the inverse metric $g^{\mu\nu}$ gives Einstein's vacuum equations,
\begin{equation} \label{eqs:EinsteinVacuum}
R_{\mu\nu}=0
\end{equation}
where $R_{\mu\nu}$ is the Ricci curvature tensor. Einstein's vacuum equations determine the evolution of the spacetime metric $g_{\mu\nu}$ in empty space.

This paper focuses on \emph{diagonal} spacetime metrics. These are metrics that can be written in the form
\begin{equation} \label{eq:DiagonalMetric}
g_{\mu\nu} = (H_\mu) ^2 \delta_{\mu\nu}
\end{equation}
where $\delta_{\mu\nu}$ is the Kroncker delta. Here and in the rest of this paper, Einstein's summation convention is \emph{not} used. In matrix form, the diagonal metric is
\begin{equation} \label{eq:DiagonalMatrix}
g_{\mu\nu}=\begin{bmatrix}(H_0) ^2 & 0 & 0 & 0 \\ 0 & (H_1) ^2 & 0 & 0 \\ 0 & 0 & (H_2) ^2 & 0 \\ 0 & 0 & 0 & (H_3) ^2 \end{bmatrix}.
\end{equation}
It is convenient not to worry about the sign of the metric. Instead, one may restore the proper metric signature $(-+++)$ by substituting $H_0 \rightarrow i H_0$.

It is a well known result that \emph{every} metric $g_{\mu\nu}$ may be diagonalized at \emph{any} given event of spacetime (e.g. by using Riemann normal coordinates) \cite{bib:Wald}. Nevertheless, this is a local result, which holds globally only for very specific spacetime metrics. This means that the class of metrics that can be casted into the diagonal form (\ref{eq:DiagonalMetric}) \emph{globally} should be expected to have \emph{unique features}. It is important to keep in mind that the \emph{diagonality} of the metric \emph{is not an invariant property}. In other words, some non-diagonal metrics $g_{\mu\nu}$ may  be transformed to the diagonal form (\ref{eq:DiagonalMetric}) by a proper choice of coordinates.

The metric (\ref{eq:DiagonalMetric}) describes a wide range of physical phenomena. In particular, it includes the Schwarzschild black hole \cite{bib:Schwarzschild}, the Kasner metric \cite{bib:Kasner}, the Friedmann-Robertson-Walker model of cosmology \cite{bib:ExactSolutions}, the Milne model of cosmology \cite{bib:ExactSolutions}, a certain class of single-polarized plane gravitational waves \cite{bib:ExactSolutions} and special cases of gravitational solitons \cite{bib:BelinskiVerdaguer}.

The goal of the rest of this paper is to study the system of vacuum Einstein Eqs. (\ref{eqs:EinsteinVacuum}) for diagonal metrics (\ref{eq:DiagonalMetric}). In particular, the mathematical structure of the equations and their physical implications on gravitational waves are emphasized.

Section \ref{sec:FieldEquations} includes a derivation of the Einstein equations in the case of the diagonal metric (\ref{eq:DiagonalMetric}). A convenient form for analyzing the equations is obtained. Section \ref{sec:Compatibility} shows that if at least one Killing vector exists, Einstein's equations for diagonal metric are compatible. In section \ref{sec:PlaneGravitationalWaves}, plane gravitational waves are studied. A simple criteria for asymptotic flatness and compatibility of the field equations for plane waves are derived. One of the most famous examples of such plane waves is the Bondi-Pirani-Robinson (BPR) waves \cite{bib:RobinsonBondiWaves}. In section \ref{sec:Stability} it is proven that such waves are stable with respect to diagonal perturbations that depend on $1+2$ coordinates.

As a concrete example, in section \ref{sec:Transparency} a BPR wave with soliton-like properties is studied. The emitted (perturbation) wave is shown to travel through the BPR wave \emph{without any reflection} and \emph{independently of the amplitude of the BPR wave}. The latter implies that a \emph{strong (BPR) gravitational wave would be transparent} to the perturbation wave. The only remnant of the collision is a phase shift which depends on the angle between the two waves. These properties, which are typically exhibited by solitons, \emph{suggest that the $1+2$ vacuum Einstein equations for diagonal metrics are integrable}, similarly to the $1+1$ vacuum Einstein equations \cite{bib:BelinskiZakharov1, bib:BelinskiZakharov2}.

\section{The field equations} \label{sec:FieldEquations}

For the diagonal metric (\ref{eq:DiagonalMetric}), the inverse metric is
\begin{equation}
g^{\mu\nu} = \frac{1}{(H_\mu) ^2} \delta^{\mu\nu}
\end{equation}
and the Christoffel symbols are
\begin{eqnarray}
\Gamma_{\mu\nu} ^\lambda &=& 0 \\ \notag
\Gamma_{\mu \nu} ^\mu &=& \partial _\nu \left(\ln H_\mu \right) \\ \notag
\Gamma_{\mu\mu} ^\nu &=& -\frac{1}{(H_\nu) ^2} H_\mu \partial_\nu H_\mu
\end{eqnarray}
where $\mu,\nu,\lambda$ are assumed to be mutually exclusive indices ($\mu\neq \nu$, $\mu \neq \lambda$, $\nu\neq \lambda$). Define the rotation coefficients
\begin{equation} \label{eq:RotationCoefficients}
Q_{\mu\nu} = \frac{1}{H_\nu} \partial_\nu H_\mu
\end{equation}
with which one can write the off-diagonal Ricci curvature tensor as
\begin{equation} \label{eq:NonDiagonalRicci}
R_{\mu \nu} = -\sum_{\lambda \neq \mu,\nu} \frac{H_\mu}{H_\lambda} \left(\partial_\nu Q_{\lambda \mu} - Q_{\lambda \nu} Q_{\nu \mu}\right)
\end{equation}
for $\mu\neq \nu$. As for the diagonal elements, the Ricci tensor gives
\begin{equation} \label{eq:DiagonalRicci}
R_{\mu\mu} = -\sum_{\nu \neq \mu} \frac{H_\mu}{H_\nu } E_{\mu\nu}
\end{equation}
where 
\begin{equation} \label{eqs:ECoefficients}
E_{\mu \nu} = \partial_\nu Q_{\mu \nu} + \partial_\mu Q_{\nu\mu} + \sum_{\lambda\neq \mu,\nu} Q_{\mu \lambda} Q_{\nu \lambda}
\end{equation}
The scalar curvature is
\begin{equation}
R=-2 \sum _{\mu<\nu} \frac{E_{\mu \nu}}{H_\mu H_\nu}.
\end{equation}
Since the determinant of the metric is $\det g=(H_0 H_1 H_2 H_3)^2$, the Einstein-Hilbert action (\ref{eq:EinsteinHilbertAction}) is
\begin{equation}
S = - i \sum_{\mu \neq \nu \neq \lambda \neq \sigma} \int E_{\mu\nu} H_\lambda H_\sigma d^4 x.
\end{equation}
If one performs an integration by parts to remove the second order derivatives in $E_{\mu\nu}$, a very concise formula for the Einstein-Hilbert action in terms of the metric coefficients only is obtained,
\begin{equation} \label{eq:DiagonalAction}
S = \sum _\mu \int \frac{i}{H_\mu} \sum_{\nu \neq \lambda \neq \sigma \neq \mu} H_\nu (\partial_\mu H_\lambda) (\partial_\mu H_\sigma) d^4 x.
\end{equation}

The reader should not be alarmed by the appearance of the imaginary root of unity $i=\sqrt{-1}$. It is there due to the signature of the metric and the transformation $H_0 \rightarrow i H_0$, which was mentioned after Eq. (\ref{eq:DiagonalMatrix}), reveals immediately that the action (\ref{eq:DiagonalAction}) is manifestly real-valued as expected.

\section{Compatibility} \label{sec:Compatibility}
When studying \emph{general} metrics, the symmetric $g_{\mu\nu}$ has ten elements, four of which may be eliminated through the use of gauge transformations. This makes the vacuum Einstein equations $R_{\mu\nu}=0$ an overdetermined system of ten equations for six unknowns. In normal circumstance this might raise the question of compatibility. Nevertheless, this is not an issue, as one can prove using the four Bianchi identities that the vacuum Einstein equations are indeed compatible \cite{bib:Wald}.

However, the situation is rather different when discussing diagonal metrics (\ref{eq:DiagonalMetric}). In this case, Einstein's vacuum equations $R_{\mu\nu}=0$ give ten equations again, but this time for only \emph{four} unknown functions $H_\mu$ ($\mu=0,1,2,3$). In this case, the usual argument using the Bianchi identities ceases to hold, and an important question thus arises:
\emph{are the Einstein's vacuum equations for diagonal metrics compatible?}

Consider the diagonal metric in Eq. (\ref{eq:DiagonalMetric}) with the additional assumption that it is independent of $x^3$. Mathematically, this means that the metric has the Killing vector $\partial_3$ and depends on the three coordinates $x^0,x^1,x^2$ only. In this case, the off-diagonal terms of the Ricci curvature tensor (\ref{eq:NonDiagonalRicci}) give only three independent equations
\begin{equation} \label{eqs:DiagonalEinstein3DOffDiagonal}
R_{01}=R_{02}=R_{12}=0
\end{equation}
coupled to the four diagonal equations
\begin{equation}\label{eqs:DiagonalEinstein3DDiagonal}
R_{\mu\mu}=0
\end{equation}
for $\mu = 0,1,2,3$. Eqs. (\ref{eqs:DiagonalEinstein3DOffDiagonal}) and (\ref{eqs:DiagonalEinstein3DDiagonal}) will be referred to as \emph{the $1+2$ vacuum Einstein equations for diagonal metrics}. The $1+2$ vacuum Einstein equations for diagonal metrics form an overdetermined system of seven equations for four unknown functions. As an overdetermined system, the compatibility of the seven equations must be proven, as it does not follow from the argument typically used for non-diagonal metrics. The authors could not find any evidence for such a result in the literature. Whether the $1+2$ Einstein equations for diagonal metrics are indeed compatible is a very natural question to ask, as such metrics have many applications in cosmology and astronomy, some of which will be described in the next sections. Fortunately, it turns out that the answer is affirmative, as the next theorem proves.

\begin{thm} \label{thm:Compatibility}
The $1+2$ Einstein equations for diagonal metrics (\ref{eqs:DiagonalEinstein3DOffDiagonal},\ref{eqs:DiagonalEinstein3DDiagonal}) are compatible.
\end{thm}
\begin{proof}
Proving the statement of this theorem using the original degrees of freedom $H_0,H_1,H_2$ and $H_3$ is rather tedious. Instead, it is much easier to exploit the special role of $H_3$ as the degree of freedom that corresponds to the Killing vector $\partial_3$. Define, 
\begin{eqnarray}
H_0 = e^{-\Lambda} \gamma && H_1 = e^{-\Lambda} \beta \hspace{.4cm} \\ \notag
H_2 = e^{-\Lambda} \alpha && H_3 = e^{\Lambda}.
\end{eqnarray}
Using the new degrees of freedom $\alpha,\beta,\gamma$ and $\Lambda$, the off-diagonal Einstein Eqs. (\ref{eqs:DiagonalEinstein3DOffDiagonal}) are
\begin{eqnarray}
\label{eqs:NonDiagonal}
\\ \partial_0 \partial_1 \alpha &=& -2\alpha (\partial_0 \Lambda) (\partial_1 \Lambda) + \frac{(\partial_0 \beta) (\partial_1 \alpha)}{\beta} + \frac{(\partial_0 \alpha) (\partial_1 \gamma)}{\gamma}  \nonumber \\
\partial_0 \partial_2 \beta &=& -2 \beta (\partial_0 \Lambda) (\partial_2 \Lambda) + \frac{(\partial_0 \alpha) (\partial_2 \beta)}{\alpha} + \frac{(\partial_0 \beta) (\partial_2 \gamma)}{\gamma} \nonumber \\
\partial_1 \partial_2 \gamma &=& -2\gamma (\partial_1 \Lambda) (\partial_2 \Lambda) + \frac{(\partial_1 \alpha) (\partial_2 \gamma)}{\alpha} + \frac{(\partial_1 \gamma) (\partial_2 \beta)}{\beta} \nonumber.
\end{eqnarray}
As for the diagonal Eqs. (\ref{eqs:DiagonalEinstein3DDiagonal}), it is convenient to represent them in an equivalent form through the variational formulation. The Lagrangian density of the Einstein-Hilbert action (\ref{eq:DiagonalAction}) is now
\begin{eqnarray} \label{eq:LagrangianDiagonal}
\mathcal{L} &=& 2 \left[\frac{\alpha \beta}{\gamma} (\partial_0 \Lambda)^2 - \frac{\alpha \gamma}{\beta} (\partial_1 \Lambda)^2 - \frac{\beta \gamma}{\alpha} (\partial_2 \Lambda)^2 \right. \\ \notag
&& \left. \hspace{0.4cm}- \frac{(\partial_0 \alpha) (\partial_0 \beta)}{\gamma} + \frac{(\partial_1 \alpha) (\partial_1 \gamma)}{\beta} + \frac{(\partial_2 \beta)(\partial_2 \gamma)}{\alpha}\right].
\end{eqnarray}
The variations $\frac{\delta S}{\delta \alpha}=\frac{\delta S}{\delta \beta}=\frac{\delta S}{\delta \gamma}=0$ give three of the diagonal equations
\begin{widetext}
\begin{eqnarray} \label{eqs:abc}
\beta \partial_0 \partial_0 \beta - \gamma \partial_1 \partial_1 \gamma &=& -\beta^2 (\partial_0 \Lambda)^2 + \gamma^2 (\partial_1 \Lambda)^2 - \frac{\beta^2 \gamma^2}{\alpha^2} (\partial_2 \Lambda)^2 + \frac{\beta}{\gamma} (\partial_0 \beta)(\partial_0 \gamma) - \frac{\gamma}{\beta} (\partial_1 \beta)(\partial_1 \gamma) + \frac{\beta \gamma}{\alpha^2} (\partial_2 \beta)(\partial_2 \gamma) \\ \notag
\alpha \partial_0 \partial_0 \alpha - \gamma \partial_2 \partial_2 \gamma &=& -\alpha^2 (\partial_0 \Lambda)^2 - \frac{\alpha^2 \gamma^2}{\beta^2} (\partial_1 \Lambda)^2 + \gamma^2 (\partial_2 \Lambda)^2 + \frac{\alpha}{\gamma} (\partial_0 \alpha)(\partial_0 \gamma) + \frac{\alpha \gamma}{\beta^2} (\partial_1 \alpha)(\partial_1 \gamma) - \frac{\gamma}{\alpha} (\partial_2 \alpha)(\partial_2 \gamma) \\ \notag
\alpha \partial_1 \partial_1 \alpha + \beta \partial_2 \partial_2 \beta &=& - \frac{\alpha^2 \beta^2}{\gamma^2} (\partial_0 \Lambda)^2 - \alpha^2 (\partial_1 \Lambda)^2 - \beta^2 (\partial_2 \Lambda)^2 + \frac{\alpha \beta}{\gamma^2} (\partial_0 \alpha)(\partial_0 \beta) + \frac{\alpha}{\beta} (\partial_1 \alpha)(\partial_1 \beta) + \frac{\beta}{\alpha} (\partial_2 \alpha)(\partial_2 \beta),
\end{eqnarray}
\end{widetext}
while the last diagonal equation, $\frac{\delta S}{\delta \Lambda}=0$ is
\begin{equation} \label{eq:Lambda}
\partial_0 \left(\frac{\alpha\beta}{\gamma} \partial_0 \Lambda\right)-\partial_1 \left(\frac{\alpha\gamma}{\beta} \partial_1 \Lambda\right)-\partial_2 \left(\frac{\beta \gamma}{\alpha} \partial_2 \Lambda \right)=0.
\end{equation}
To prove the statement of the theorem, we differentiate each of Eqs. (\ref{eqs:abc}) with respect to $x^2, x^1$ and $x^0$ respectively. This gives three third order equations for $\alpha, \beta$ and $\gamma$. One may now eliminate each of the third order terms using the non-diagonal Eqs. (\ref{eqs:NonDiagonal}). After a lengthy algebra, one sees that with the aid of Eqs. (\ref{eqs:NonDiagonal}) once more, all $38$ terms in each equation completely vanish. Therefore the $1+2$ Einstein equations for diagonal metrics are indeed compatible.
\end{proof}


The degrees of freedom $\alpha,\beta,\gamma$, and $\Lambda$ from theorem (\ref{thm:Compatibility}) are very useful. They provide an alternative way to study general diagonal spacetime metrics (\ref{eq:DiagonalMetric}). With such degrees of freedom the spacetime interval is
\begin{equation} \label{eq:SpacetimeInterval1P2}
ds^2 = e^{-2\Lambda} \left[-(\gamma dx^0)^2 + (\beta dx^1)^2 + (\alpha dx^2)^2\right] + e^{2\Lambda} (dx^3)^3
\end{equation}

The spacetime interval (\ref{eq:SpacetimeInterval1P2}) naturally generalizes the interval studied in \cite{bib:BelinskiZakharov1,bib:BelinskiZakharov2}. To see this, assume the metric is independent of $x^2$, set $\beta=\gamma$ and define $f=\gamma^2 e^{-2 \Lambda}$. This turns the spacetime interval (\ref{eq:SpacetimeInterval1P2}) into
\begin{equation} \label{eq:MetricZB}
ds^2 = f \left[-(dx^0)^2 + (dx^1)^2\right] + \alpha^2 e^{-2\Lambda} (dx^2)^2 + e^{2\Lambda} (dx^3)^3
\end{equation}
For this metric one may use the inverse scattering transform \cite{bib:BelinskiZakharov1,bib:BelinskiZakharov2} to derive gravitational solitons on diagonal metrics \cite{bib:BelinskiVerdaguer}. 

There is another merit of using the new degrees of freedom, as in the course of the proof of theorem (\ref{thm:Compatibility}) we just derived \emph{a new conservation law}. This is of course Eq. (\ref{eq:Lambda}). Whenever the metric is asymptotically flat, it also yields the integral of motion,
\begin{equation}
P=\int \left[\frac{\alpha\beta}{\gamma} \partial_0 \Lambda \right] dx^1 dx^2
\end{equation}
which is the conjugate momentum of the function $\Lambda$, as can be easily seen from the Lagrangian in Eq. (\ref{eq:LagrangianDiagonal}).

\section{Plane gravitational waves} \label{sec:PlaneGravitationalWaves}

Waves come in many forms and shapes. The simplest of which are of course plane waves, whose wavefronts are parallel planes extended ad infinitum. In general relativity, plane gravitational waves are typically studied as a special case of the famous $pp$-waves \cite{bib:Griffiths}. The $pp$-class consists of any spacetime metric that can be casted into the form,
\begin{equation} \label{eq:ppWaves}
ds^2 = H(u,x,y) du^2 + 2du dv + dx^2 +dy^2.
\end{equation}
Recently, a coordinate-free definition of them was given \cite{bib:Steele}. For such a metric, Einstein's vacuum equation reduces to Laplace's equation,
\begin{equation}
\frac{\partial^2 H}{\partial x^2} + \frac{\partial^2 H}{\partial y^2} = 0
\end{equation}
and is therefore \emph{linear} in $H$. A $pp$-wave is called a \emph{plane wave} if $H$ can be transformed into
\begin{equation}
H(u,x,y)=a(u) (x^2-y^2) +2b(u)xy
\end{equation}
where $a(u)$ and $b(u)$ control the waveform of the two possible polarizations.

Diagonal metrics of $1+2$ coordinates with the spacetime interval (\ref{eq:SpacetimeInterval1P2}) may describe $pp$ spacetimes as well as many non-$pp$ spacetimes (such as the Schwarzschild black hole). Either way, it is important to keep in mind that they correspond to a different class of solutions of Einstein's vacuum equations.

The equations considered in the previous section allow investigating diagonal metrics (\ref{eq:SpacetimeInterval1P2}) with 
\begin{eqnarray} \label{eqs:NaivePlaneWaves}
\Lambda=\Lambda(\eta) \hspace{0.7 cm} \alpha = \alpha (\eta) \\ \notag
\beta = \beta (\eta) \hspace{0.7 cm} \gamma=\gamma(\eta)
\end{eqnarray}
where each metric coefficient depends on all three coordinates through
\begin{equation} \label{eq:Phase}
\eta=\frac{1}{2}(-x^0 + p x^1 + q x^2).
\end{equation}
The form of the parameter $\eta$ corresponds to the naive definition of a plane wave propagating with velocity $v=(p^2 + q^2)^{-1}$, similarly to plane waves studied in other fields of physics. The factor $1/2$ is there because then setting $p=1$ and $q=0$ reduces $\eta$ to its former definition from the theory of gravitational solitons \cite{bib:BelinskiZakharov1}. But is this `naive' plane wave consistent with the canonical definition of a plane wave as a subset of the $pp$-class?

To answer this question, consider Einstein's vacuum equations. For the metric coefficients (\ref{eqs:NaivePlaneWaves}), the diagonal Eq. (\ref{eq:Lambda}) gives
\begin{equation} \label{eq:LambdaPlane}
\left(\frac{\alpha \beta}{\gamma} \Lambda'\right)' = p^2 \left(\frac{\alpha \gamma}{\beta} \Lambda'\right)' + q^2 \left(\frac{\beta \gamma}{\alpha} \Lambda'\right)'
\end{equation}
where the prime denotes differentiation with respect to $\eta$. Integrating it and noticing that $\Lambda'$ must vanish at some moment of time shows that the integration constant is trivial. Hence the following algebraic relation holds,
\begin{equation} \label{eq:PlaneWaveAlgebraic}
\alpha^2 \beta^2 = p^2 \alpha^2 \gamma^2 + q^2 \beta^2 \gamma^2.
\end{equation}
To ensure asymptotic flatness, one \emph{may} impose the conditions $\alpha,\beta,\gamma \rightarrow 1$ and $\Lambda\rightarrow \text{constant}$ at spatial infinity. In particular, this guarantees that the metric converges to the Minkowski metric. Taking this limit in Eq. (\ref{eq:PlaneWaveAlgebraic}) gives the relation
\begin{equation} \label{eq:PlaneWaveSpeedOfLight}
1=p^2+q^2,
\end{equation}
which has a clear physical interpretation: \emph{a sufficient condition for the gravitational wave to be asymptotically flat is for it to travel precisely at the speed of light}, just like the plane $pp$-waves \cite{bib:Griffiths}. 

Whenever $p,q\neq 0$ the off-diagonal Eqs. (\ref{eqs:NonDiagonal}) give three nonlinear ordinary differential equations,
\begin{eqnarray} \label{eqs:NonDiagonalPlane}
\frac{\alpha''}{\alpha} &=& -2 (\Lambda')^2 + \frac{\alpha'}{\alpha} \left(\frac{\beta'}{\beta}+ \frac{\gamma'}{\gamma}\right) \\ \notag
\frac{\beta''}{\beta} &=& -2 (\Lambda')^2 + \frac{\beta'}{\beta} \left(\frac{\alpha'}{\alpha}+ \frac{\gamma'}{\gamma}\right) \\ \notag
\frac{\gamma''}{\gamma} &=& -2 (\Lambda')^2 + \frac{\gamma'}{\gamma} \left(\frac{\alpha'}{\alpha}+ \frac{\beta'}{\beta}\right).
\end{eqnarray}
Similarly, the diagonal Eqs. (\ref{eqs:abc}) are
\begin{eqnarray} \label{eqs:DiagonalPlane}
\\ \alpha^2 (\beta \beta'' - p^2 \gamma \gamma'') &=& (-\alpha^2 \beta^2 + p^2 \alpha^2 \gamma^2 - q^2 \beta^2 \gamma^2) (\Lambda')^2 \nonumber \\
&& + \left(\frac{\alpha^2 \beta}{\gamma} - p^2 \frac{\alpha^2 \gamma}{\beta} + q^2 \beta \gamma\right) \beta' \gamma' \nonumber \\
\beta^2 (\alpha \alpha'' - q^2 \gamma \gamma'') &=& (-\alpha^2 \beta^2 - p^2 \alpha^2 \gamma^2 + q^2 \beta^2 \gamma^2) (\Lambda')^2 \nonumber \\
&& + \left(\frac{\alpha \beta^2}{\gamma} + p^2 \alpha \gamma - q^2 \frac{\beta^2 \gamma}{\alpha}\right) \alpha' \gamma' \nonumber \\
\gamma^2 (p^2 \alpha \alpha'' + q^2 \beta \beta'') &=& (-\alpha^2 \beta^2 - p^2 \alpha^2 \gamma^2 - q^2 \beta^2 \gamma^2) (\Lambda')^2 \nonumber \\
&&+ \left(\alpha \beta + p^2 \frac{\alpha \gamma^2}{\beta} + q^2 \frac{\beta \gamma^2}{\alpha} \right) \alpha' \beta' \nonumber.
\end{eqnarray}

As in the last section, Einstein's equations in this case are an overdetermined system. They are seven equations for four unknowns $\alpha,\beta,\gamma$ and $\Lambda$. Unfortunately, theorem (\ref{thm:Compatibility}) is no longer valid, as the functions sought here are of a very special form, depending on the coordinates $x^0,x^1,x^2$ through the phase $\eta$ only. This means that compatibility has to be studied once more. 

We substitute then the off-diagonal Eqs. (\ref{eqs:NonDiagonalPlane}) in the diagonal Eqs. (\ref{eqs:DiagonalPlane}). A lengthy algebra that exploits the relation just derived in Eq. (\ref{eq:PlaneWaveAlgebraic}) reveals that the plane wave Eqs. (\ref{eqs:NonDiagonalPlane}) and (\ref{eqs:DiagonalPlane}) are compatible if and only if,
\begin{equation} \label{eq:PlaneWaveCompatibility}
(\alpha^2)' (\beta^2)' = p^2 (\alpha^2)' (\gamma^2)' + q^2 (\beta^2)' (\gamma^2)'.
\end{equation}
Indeed, this means that every gravitational wave that is diagonal and \emph{planar} must satisfy the compatibility condition (\ref{eq:PlaneWaveCompatibility}). In virtue of Eq. (\ref{eq:PlaneWaveSpeedOfLight}) we will also assume that it propagates at the speed of light. It will next be proven that one may always take $p=1$ and $q=0$. Thus, without loss of generality it is sufficient to consider a wave propagating along the positive $x^1$-axis only.

To prove this claim, subtract each pair of consecutive off-diagonal equations (\ref{eqs:NonDiagonalPlane}) to write them as
\begin{eqnarray}
(\ln \frac{\alpha}{\beta})'' &=& (\ln \frac{\alpha}{\beta})' (\ln \frac{\gamma}{\alpha\beta})' \\ \notag
(\ln \frac{\gamma}{\alpha})'' &=& (\ln \frac{\gamma}{\alpha})' (\ln \frac{\beta}{\alpha\gamma})' \\ \notag
(\ln \frac{\beta}{\gamma})'' &=& (\ln \frac{\beta}{\gamma})' (\ln \frac{\alpha}{\beta \gamma})'.
\end{eqnarray}
These equations can be integrated immediately to yield
\begin{eqnarray} \label{eqs:NonDiagonalPlaneIntegrated}
(\ln \frac{\alpha}{\beta})' &=& C_1 \frac{\gamma}{\alpha \beta} \\ \notag
(\ln \frac{\gamma}{\alpha})' &=& C_2 \frac{\beta}{\alpha \gamma} \\ \notag
(\ln \frac{\beta}{\gamma})' &=& C_3 \frac{\alpha}{\beta \gamma}
\end{eqnarray}
where $C_1, C_2$ and $C_3$ are arbitrary constants of integration. Adding Eqs. (\ref{eqs:NonDiagonalPlaneIntegrated}) and multiplying the result by $\alpha\beta\gamma$ gives a compatibility condition
\begin{equation} \label{eq:PlaneCompatibility}
C_1 \gamma ^2 + C_2 \beta^2 + C_3 \alpha^2 = 0.
\end{equation}
Therefore, there are several possibilities.

If $C_1=0$  then $\beta$ is proportional to $\alpha$. Rescaling the coordinate $x^1$ by the same proportionality factor shows that such metric is of the form (\ref{eq:SpacetimeInterval1P2}) with $\beta=\alpha$. The second case is when either $C_2=0$ or $C_3=0$. Assuming without loss of generality that it is the former $C_2=0$, show that $\gamma$ is proportional to $\alpha$ and therefore from Eq. (\ref{eq:PlaneWaveCompatibility}), $\beta=\pm \alpha$ giving again the form (\ref{eq:SpacetimeInterval1P2}) with $\beta=\alpha$. The last case is when $\alpha$, $\beta$ and $\gamma$ are all proportional to one another. 


We have therefore shown that without loss of generality, the metric can \emph{always} be written in the form (\ref{eq:SpacetimeInterval1P2}) with $\alpha=\beta$. In this work we study the case where $p=1$ and $q=0$, and the variable $\eta$ in Eq. (\ref{eq:Phase}) reduces back to its definition from \cite{bib:BelinskiZakharov1}. This means that all of the metric coefficients depend only on $\eta=\frac{1}{2}(x^1-x^0)$. However, then Eq. (\ref{eq:PlaneWaveCompatibility}) together with Eq. (\ref{eq:PlaneWaveSpeedOfLight}) prove that $\alpha=\gamma$ as well. Further rescaling of $\eta$ finally yields the metric
\begin{equation} \label{eq:MetricBPR}
ds^2 = e^{-2\Lambda} \left[-(dx^0)^2 + (dx^1)^2 + (\alpha dx^2)^2\right] + e^{2\Lambda} (dx^3)^3.
\end{equation}
This is equivalent to the spacetime interval (\ref{eq:SpacetimeInterval1P2}) with $\beta=\gamma=1$, $p=1$ and $q=0$. For such a metric, all of Einstein's vacuum Eqs. (\ref{eqs:NonDiagonal}), (\ref{eqs:abc}) and (\ref{eq:Lambda}) yield a single equation,
\begin{equation} \label{eq:BPR1}
\alpha'' + 2\alpha (\Lambda')^2 = 0.
\end{equation}
The solutions of Eq. (\ref{eq:BPR1}) are the famous \emph{Bondi-Pirani-Robinson (BPR) waves} \cite{bib:RobinsonBondiWaves}. It is a known fact that a BPR wave is in particular a $pp$-wave \cite{bib:Griffiths}.


\section{Stability} \label{sec:Stability}
As shown in the last section, by properly choosing the coordinate system used, one may always describe a plane gravitational wave with a diagonal metric as propagating along the positive $x^1$-axis. This means that all the coefficients of the metric defined in Eq. (\ref{eq:SpacetimeInterval1P2}) are functions of the light-cone coordinate $\eta=\frac{1}{2}(x^1-x^0)$ alone. Furthermore, it was shown in the last section that one may also assume that $\beta=\gamma=1$. These waves satisfy the Bondi-Pirani-Robinson Eq. (\ref{eq:BPR1}),
\begin{equation} \label{eq:BPR}
0=\alpha_0 '' + 2\alpha_0 (\Lambda_0 ')^2
\end{equation}
where $\alpha_0 = \alpha(\eta)$ and $\Lambda_0 = \Lambda(\eta)$, and the prime denotes differentiation with respect to the variable $\eta$.

\begin{figure}[h!]
\begin{center}
\includegraphics[scale=0.8]{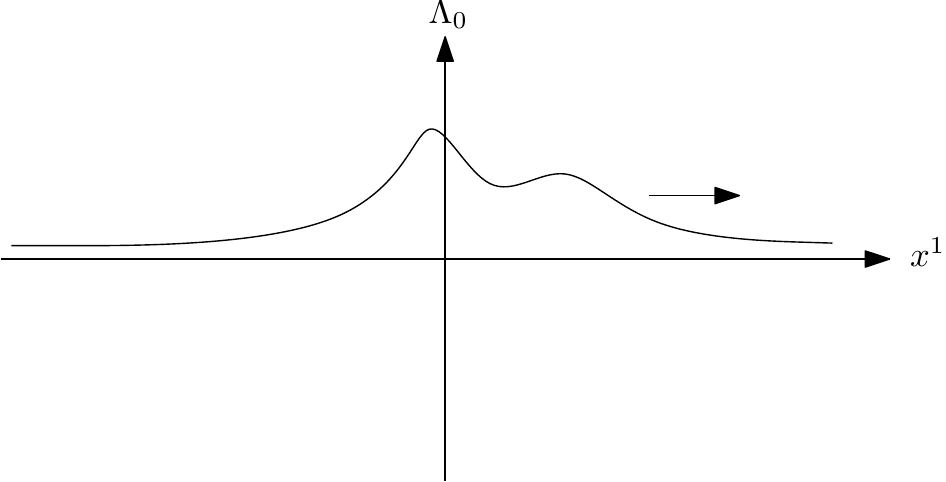}
\caption[The profile of the Bondi-Pirani-Robinson gravitational wave]{The profile of the Bondi-Pirani-Robinson gravitational wave can have arbitrary form, with the wave propagating in the positive $x^1$-direction only.
  	\label{fig:BPRWaveProfile}}
\end{center}
\end{figure}

Eq. (\ref{eq:BPR}) reveals a peculiar situation where the two degrees of freedom satisfy only one equation, which is thus \emph{an underdetermined system} for $\alpha_0$ and $\Lambda_0$. Therefore one of the functions $\alpha_0$ or $\Lambda_0$ can be set arbitrary. Physically, this means that such a gravitational wave may have any wave profile as determined by $\Lambda_0$ (see figure \ref{fig:BPRWaveProfile}).

As a practical example, consider the solution
\begin{equation} \label{eq:BargmannSolution}
\alpha_0=\tanh \eta.
\end{equation}
Solving Eq. (\ref{eq:BPR}) for $\Lambda_0'$ implies that $\Lambda_0'=\frac{\pm 1}{\cosh \eta}$ is of a soliton-like form. Despite being only an example, this solution is of fundamental importance in scattering theory. When $\Lambda_0'=\frac{\pm 1}{\cosh \eta}$ we see that $(\Lambda_0 ')^2 = \frac{d^2}{d\eta^2} \ln (1+e^{-2\eta})= \frac{1}{\cosh^2 \eta}$. Therefore, $(\Lambda_0 ')^2$ is a \emph{Bargmann potential} \cite{bib:TheoryOfSolitons}. It is worthwhile reviewing the general definition of Bargmann potentials as their unique characteristics appear in this stability problem.

In general, a Bargmann potential is a function of the form
\begin{equation} \label{eq:BargmannPotential}
(\Lambda_0 ')^2 = \frac{d^2}{d\eta^2} \ln \Delta
\end{equation}
where $\Delta$ is the determinant
\begin{equation}
\Delta = \det \left[\delta _{ij} + \frac{M_i ^2 e^{-(\lambda_i + \lambda_j) \eta}}{\lambda_i + \lambda_j}\right],
\end{equation}
$i,j = 1,2,\dots,n$, the constants $M_i$ are real, and $\lambda_i>0$. In the case considered above, we have $N=1$, $\lambda_1 = 1$ and $M_1 ^2 = 2$. Bargmann potentials are also commonly called "the $N$-solitonic potentials" or "\emph{reflectionless} potentials".

The solution $\alpha_0$ in Eq. (\ref{eq:BargmannSolution}) vanishes at $\eta=0$, and therefore the metric (\ref{eq:MetricBPR}) is singular at this point. Nevertheless, this is not a physical singularity. A simple but lengthy calculation shows that all of the components of the Weyl tensor $C_{\mu\nu\lambda\sigma}$ \cite{bib:Wald} are proportional to either one of the two components $C_{0220}$ or $C_{0330}$ given by 
\begin{equation}
C_{0220} = (\alpha_0) ^2 e^{-4 \Lambda_0} C_{0330} = -(\alpha_0) ^2 e^{-2 \Lambda_0} [\Lambda_0 '' + 3 (\Lambda_0 ')^2].
\end{equation}
Therefore the BPR spacetime is never singular when $\alpha_0=0$. In fact, $\alpha_0=0$ corresponds to events at which the spacetime is flat.

As mentioned earlier, the metric (\ref{eq:MetricBPR}) will be asymptotically flat if $\Lambda_0 \rightarrow 0$ as $|\eta| \rightarrow  \infty$. Given a particular $\Lambda_0$ satisfying this property, the general solution $\alpha_0$ of Eq. (\ref{eq:BPR}) may be unbounded. If one further imposes the conditions $\alpha_0\rightarrow const$ and $\alpha_0' \rightarrow 0$ as $\eta\rightarrow -\infty$, then $\alpha_0 \sim c_1 + c_2 \eta$ at infinity $\eta\rightarrow\infty$, for two real constants $c_1$ and $c_2$. Most of the rest of this work will focus on such solutions $\alpha_0$ for which $c_2=0$. This is the class of functions $\alpha_0$ that are \emph{bounded at infinity}. In particular, this class of solutions includes the Bargmann potential solution given by Eq. (\ref{eq:BargmannSolution}).





Since the BPR waves are a solution the Einstein vacuum Eq. (\ref{eqs:EinsteinVacuum}), the stress-energy tensor $T_{\mu\nu}$ is identically zero and there is no other source of gravitational fields. One can therefore think of the field and the stress it produces as being in equilibrium under the gravitational effects of the Bondi-Pirani-Robinson gravitational waves themselves. We have an equilibrium, but \emph{is it stable}? The goal of this section is to explore this fascinating question.

The stability of plane gravitational waves is of great importance. In the past two decades there have been serious attempts to detect gravitational waves, so far unsuccessful. This includes the experiment conducted at the Laser Interferometer Gravitational-Wave Observatory (LIGO), a monumental project costing several hundred million dollars \cite{bib:LIGO}. If plane gravitational waves are unstable, there is a very good reason they are difficult to detect, as small departures from idealized waves might destroy them.

To study this question of stability, consider a small perturbation of the Bondi-Pirani-Robinson waves that allows them to propagate weakly in the perpendicular $x^2$-direction and to even reflect in the negative $x^1$-direction. It is convenient to use light-cone coordinates \cite{bib:BelinskiZakharov1},
\begin{equation} \label{eq:EtaZeta}
\eta=\frac{1}{2}(x^1-x^0) \hspace{1cm} \zeta=\frac{1}{2}(x^1+x^0).
\end{equation}
In the approximation of a small perturbation, the equations are linear in first order. It is then possible to separate their disturbance into proper modes and find their frequencies, whether real (stability) or imaginary (instability). Therefore we consider metric coefficients of the form
\begin{eqnarray}
\alpha = \alpha_0 + \delta \alpha && \beta = 1 + \delta \beta  \\ \notag
\gamma = 1 + \delta \gamma && \Lambda = \Lambda_0 + \delta \Lambda
\end{eqnarray}
where $\delta \alpha$, $\delta \beta$, $\delta \gamma$ and $\delta \Lambda$ are small corrections that depend on all three variables $\zeta$, $\eta$ and $x^2$. Here $\alpha_0$ and $\Lambda_0$ are the original BPR coefficients mentioned in Eq. (\ref{eq:BPR}), and are dependent on $\eta$ alone: $\alpha_0 = \alpha_0(\eta)$ and $\Lambda_0= \Lambda_0(\eta)$.

Linearizing Eqs. (\ref{eqs:abc}) with respect to the perturbation, and transforming to Fourier modes $\zeta \rightarrow \Omega$ and $x^2\rightarrow k$ gives three second order differential equations for the perturbations of $\alpha, \beta$ and $\gamma$:
\begin{eqnarray} \label{eqs:BPRStabilityOffDiagonal}
\\ (\delta \beta)'' - (\delta \gamma)''&=& 4i \Omega \Lambda_0 ' (\delta\Lambda) + \left[2i\Omega \partial_\eta + \Omega^2- 2(\Lambda_0 ')^2\right](\delta\beta) \nonumber \\
&& \hspace{-0.35cm} + \left[2i\Omega \partial_\eta - \Omega^2+2(\Lambda_0 ')^2\right](\delta\gamma) \nonumber \\
(\delta\alpha)''&=& - 4\alpha_0 \Lambda_0 ' (\delta\Lambda)'  + 2 \alpha_0 (\Lambda_0')^2 (\delta\beta) \nonumber \\ 
&& + \left[2 i \Omega \partial_\eta + \Omega^2 - 4 (\Lambda_0')^2  - \frac{\alpha_0 ''}{\alpha_0}\right] (\delta\alpha) \nonumber \\
&& + \left[2 \alpha_0 ' \partial_\eta - \frac{4 k^2}{\alpha_0} - 2 \alpha_0 (\Lambda_0 ')^2\right] (\delta\gamma) \nonumber \\
(\delta\alpha)''&=& - 4\alpha_0 \Lambda_0 ' (\delta\Lambda)' + 2 \alpha_0 (\Lambda_0')^2 (\delta\gamma)  \nonumber \\
&& + (-2 i \Omega \partial_\eta + \Omega^2- 4 (\Lambda_0')^2 - \frac{\alpha_0 ''}{\alpha_0}) (\delta\alpha) \nonumber \\
&& + \left[2 \alpha_0 ' \partial_\eta + \frac{4 k^2}{\alpha_0}-2\alpha_0 (\Lambda_0')^2\right](\delta\beta) \nonumber
\end{eqnarray}
where the primes denote differentiation with respect to the variable $\eta$.

As for the off-diagonal equations in (\ref{eqs:NonDiagonal}), the first equation is a second order ordinary differential equation for the function $\delta\alpha$,
\begin{eqnarray} \label{eq:BPRStabilityDiagonal1}
(\delta\alpha)''&=&- 4\alpha_0 \Lambda_0 ' (\delta\Lambda)' -\left[\Omega^2 + 2 (\Lambda_0 ')^2\right] (\delta\alpha) \nonumber \\
&& + \alpha_0 ' \left[(\delta\beta+\delta\gamma)'-i\Omega(\delta\beta-\delta\gamma)\right].
\end{eqnarray}
The second and third equations in (\ref{eqs:NonDiagonal}) yield two first order equations for the perturbations $\delta\beta$ and $\delta\gamma$ provided that the frequency $k$ is nonzero,
\begin{eqnarray} \label{eqs:BPRStabilityDiagonal23}
(\delta\beta)'&=& \left[\frac{\alpha_0 '}{\alpha_0} + i\Omega\right](\delta\beta) - 2\Lambda_0' (\delta\Lambda) \\ \notag
(\delta\gamma)'&=& \left[\frac{\alpha_0 '}{\alpha_0} - i\Omega\right](\delta\gamma) - 2\Lambda_0' (\delta\Lambda)
\end{eqnarray}
and are otherwise automatically satisfied. Last but not least, Eq. (\ref{eq:Lambda}) is a first order equation for the difference in perturbations of $\beta$ and $\gamma$,
\begin{eqnarray} \label{eq:BPRStabilityDiagonal4}
0&=& - \left[2 i \Omega \alpha_0\partial_\eta + i \Omega \alpha_0 '- \frac{2k^2}{\alpha_0}\right](\delta\Lambda) - i \Omega \Lambda_0' (\delta\alpha) \nonumber \\
&& + \left[\alpha_0 \Lambda_0 ' (\delta\beta-\delta\gamma) \right]'
\end{eqnarray}

Thus, we obtained an overdetermined system of ODEs that includes seven equations for the four unknowns $\delta \alpha$, $\delta\beta$, $\delta\gamma$ and $\delta\Lambda$. Quite remarkably, a very similar situation occurs in what physically seems to be a completely different stability problem in general relativity - the stability of the Schwarzschild singularity. In a famous work of Regge and Wheeler \cite{bib:ReggeWheeler}, they showed that the complete set of Einstein's equation gives only seven equations for four unknowns, just like we have here. To analyze the equations herein obtained, it is convenient to distinguish two cases that are physically very different.

\subsection*{The case $k\neq0$}
describes a perturbation wave that travels in the $x^1$ as well as the $x^2$ direction. Moreover, if $\Omega\neq 0$ such a wave reflects in the negative $x^1$-direction due to its collision with the BPR wave.

Here the stability of the BPR wave and that of the Schwarzschild singularity reveals a great similarity. For the latter, Regge and Wheeler proved that the seven stability equations are equivalent to three differential equations coupled to one algebraic relation, from which the stability of the Schwarzschild black hole followed \cite{bib:ReggeWheeler}. Considering the fact that the two problems are physically quite different, it is quite surprising that also in this problem one may reduce the overdetermined set of seven equations to a much simpler set of three ordinary differential equations coupled to one algebraic relation.

\begin{thm} \label{thm:BPRStabilityEqs}
If the frequency $k\neq0$ then the seven stability equations (\ref{eqs:BPRStabilityOffDiagonal}), (\ref{eq:BPRStabilityDiagonal1}), (\ref{eqs:BPRStabilityDiagonal23}) and (\ref{eq:BPRStabilityDiagonal4}) are equivalent to four equations, three of which are first order equations for $\delta\beta$, $\delta\gamma$ and $\delta\Lambda$,
\begin{eqnarray} \label{eqs:BPRStability1stOrder}
(\delta\beta)' &=& \left[\frac{\alpha_0 '}{\alpha_0} + i \Omega\right] (\delta\beta) - 2\Lambda_0 ' (\delta\Lambda) \\ \notag
(\delta\gamma)' &=& \left[\frac{\alpha_0 '}{\alpha_0} - i \Omega\right] (\delta\gamma) - 2\Lambda_0 ' (\delta\Lambda) \\ \notag
(\delta\Lambda)' &=& -\left(\frac{i k^2}{\Omega \alpha_0 ^2} + \frac{\alpha_0 '}{2\alpha_0}\right)(\delta\Lambda) \\ \notag
&& + \left[-\frac{i}{2\Omega} \Lambda_0 '' + \frac{k^2}{2\Omega^2} \frac{\Lambda_0 '}{\alpha_0 ^2} - \frac{3i}{4\Omega}  \Lambda_0 ' \frac{\alpha_0 '}{\alpha_0}\right](\delta\beta-\delta\gamma) \\ \notag
&& + \frac{1}{2} \Lambda_0' (\delta\beta+\delta\gamma),
\end{eqnarray}
plus an algebraic relation for $\delta\alpha$,
\begin{equation} \label{eq:BPRStabilityAlgebraic}
0=(\delta\alpha)+\left(\frac{i}{2 \Omega} \alpha_0 ' + \frac{k^2}{\Omega^2 \alpha_0}\right)(\delta\beta-\delta\gamma)
\end{equation}
\end{thm}
\begin{proof}
Eqs. (\ref{eqs:BPRStabilityDiagonal23}) allow completely eliminating any derivatives of $\delta\beta$ and $\delta\gamma$ in each of the stability equation. In fact, differentiating them and plugging their derivatives into the first off-diagonal equation in (\ref{eqs:BPRStabilityOffDiagonal}) gives a trivial result. Therefore the first off-diagonal equation in (\ref{eqs:BPRStabilityOffDiagonal}) is the consequence of two of the diagonal equations.

Similarly, the diagonal Eq. (\ref{eq:BPRStabilityDiagonal1}) can be used to eliminate the second derivative of $\delta\alpha$ in the two remaining off-diagonal Eqs. (\ref{eqs:BPRStabilityOffDiagonal}). The result can be simplified even more by excluding the derivatives of $\delta\beta$ and $\delta\gamma$ again. This gives two first order equations for $\delta a$,
\begin{eqnarray}
0&=& 2\Omega (\Omega+ i \partial_\eta) (\delta\alpha) \\ \notag
&& + \left[2\alpha_0 (\Lambda_0 ')^2 - \frac{(\alpha_0 ')^2}{\alpha_0}\right] (\delta\beta) \\ \notag
&& -\left[2i\Omega \alpha_0 ' + \frac{4k^2}{\alpha_0}+2\alpha_0 (\Lambda_0 ')^2 - \frac{(\alpha_0 ')^2}{\alpha_0}\right] (\delta\gamma) \\ \notag
0&=& 2\Omega (\Omega-i \partial_\eta) (\delta\alpha) \\ \notag
&& + \left[2i\Omega \alpha_0 ' + \frac{4k^2}{\alpha_0}-2\alpha_0 (\Lambda_0 ')^2 + \frac{(\alpha_0 ')^2}{\alpha_0}\right] (\delta\beta) \\ \notag
&& + \left[2\alpha_0 (\Lambda_0 ')^2 - \frac{(\alpha_0 ')^2}{\alpha_0}\right] (\delta\gamma).
\end{eqnarray}
Adding these equations proves that the algebraic Eq. (\ref{eq:BPRStabilityAlgebraic}) holds. It can be used in the last diagonal equation (\ref{eq:BPRStabilityDiagonal4}) to eliminate $\delta\alpha$ altogether, yielding an equation for $\delta\Lambda$. This is the last equation in (\ref{eqs:BPRStability1stOrder}) (the first two equations in (\ref{eqs:BPRStability1stOrder}) were already derived in Eqs. (\ref{eqs:BPRStabilityDiagonal23})). This proves the that seven equations imply the four equations (\ref{eqs:BPRStability1stOrder}) and (\ref{eq:BPRStabilityAlgebraic}). 

Conversely, it is a straight-forward but an elaborated task to use Eqs. (\ref{eqs:BPRStability1stOrder}) and (\ref{eq:BPRStabilityAlgebraic}) in the seven Eqs. (\ref{eqs:BPRStabilityOffDiagonal}), (\ref{eq:BPRStabilityDiagonal1}), (\ref{eqs:BPRStabilityDiagonal23}) and (\ref{eq:BPRStabilityDiagonal4}) and see that they are satisfied.
\end{proof}

The fact that Eqs. (\ref{eqs:BPRStability1stOrder}) are decoupled from $\delta\alpha$, shows that one only needs to focus on these equations. Once solved, they can be immedietly used in Eq. (\ref{eq:BPRStabilityAlgebraic}) to give the function $\delta\alpha$.

Let $\delta\beta = \frac{1}{2}\alpha_0 (B^{+}+B^{-})$, $\delta\gamma=\frac{1}{2} \alpha_0 (B^{+}-B^{-})$ and $\delta\Lambda=\alpha_0 L$. From Eqs. (\ref{eqs:BPRStability1stOrder}), the functions $B^+$, $B^-$ and $L$ satisfy
\begin{eqnarray} \label{eqs:BPRStability1st2}
(B^{+})' &=& i \Omega B^{-} - 4\Lambda_0 ' L \\ \notag
(B^{-})' &=& i \Omega B^{+} \\ \notag
L' &=& -\left(\frac{i k^2}{\Omega \alpha_0 ^2} + \frac{3 \alpha_0 '}{2\alpha_0}\right) L\\ \notag
&& + \left[-\frac{i}{2\Omega} \Lambda_0 '' + \frac{k^2}{2\Omega^2} \frac{\Lambda_0 '}{\alpha_0 ^2} - \frac{3i}{4\Omega}  \Lambda_0 ' \frac{\alpha_0 '}{\alpha_0}\right] B^{-} + \frac{1}{2} \Lambda_0' B^{+}
\end{eqnarray}
One may eliminate $B^{+}$ from the third equation using the second equation, to get a first order equation relating $L$ and $B^{-}$. It is a miracle that it can be written in a very simple form
\begin{equation}
\Psi' + \left(\frac{ik^2}{\Omega \alpha_0 ^2} + \frac{3}{2} \frac{\alpha_0 '}{\alpha_0}\right) \Psi = 0
\end{equation}
where $\Psi=2 i \Omega L - \Lambda_0 ' B^{-}$ is a complex-valued `wave-function'.
It may be integrated immediately to obtain an algebraic relation between $L$ and $B^{-}$,
\begin{equation} \label{eq:PsiReflection}
\Psi \equiv 2 i \Omega L - \Lambda_0 ' B^{-} = \frac{K}{|\alpha_0| ^{3/2}} \exp \left[- \frac{ik^2}{\Omega} \int \frac{d\eta}{\alpha_0 ^2} \right]
\end{equation}
where $K=K(\Omega,k)$ is a constant of integration. This allows to obtain a single second order equation for $B^{-}$ alone.

Differentiate the second equation in (\ref{eqs:BPRStability1st2}). The derivative $(B^{+})'$ may be eliminated through the first equation in (\ref{eqs:BPRStability1st2}) while the function $L$ can also be excluded using the algebraic relation just derived. This yields a single second order equation for the function $B^{-}$ only,
\begin{equation} \label{eq:BMinus}
(B^{-})''+(\Omega^2 + 2 \Lambda_0 '^2) B^{-} + 2K \frac{\Lambda_0 '}{|\alpha_0| ^{3/2}} \exp\left[- \frac{ik^2}{\Omega} \int \frac{d\eta}{\alpha_0 ^2}\right]=0.
\end{equation}
Once solved, $L$ and $B^+$ can be easily obtained from Eq. (\ref{eq:PsiReflection}) and the second equation in (\ref{eqs:BPRStability1st2}).

The reader may be concerned of the division by $|\alpha_0|^{3/2}$ in Eqs. (\ref{eq:PsiReflection}) and (\ref{eq:BMinus}).The functions $\Psi$ and $B^{-}$ seem to be singular when $\alpha_0$ vanishes. However, the authors computed the curvature (Petrov) invariants \cite{bib:LL}, from which it is evident that the points where $\alpha_0$ vanishes are not physical singularities.


\subsection*{The case $k=0$}
corresponds to a metric which is independent of $x^2$. Physically, it represents a perturbation wave propagating along the negative $x^1$ axis towards a head on collision with the BPR wave.

In this case one may assume without loss of generality that $\beta = \gamma$, and particularly $\delta\beta = \delta\gamma$ \cite{bib:BelinskiVerdaguer}. This case $k=0$ can be naturally studied using a limiting procedure from the case $k\neq 0$. Indeed, taking the limit $k \rightarrow 0$ in Eq. (\ref{eq:PsiReflection}) yields 
\begin{equation}
\Psi = \frac{K(\Omega,0)}{|\alpha_0| ^{3/2}}.
\end{equation}
One can see that the $\psi$ is completely independent of $\eta$. This leads to an astonishing fact. A gravitational wave of small amplitude traveling along the negative $x^1$-axis will go straight through the BPR wave which is traveling in the opposite direction. This result is independent of the amplitude of the BPR wave. In other words, in this case the perturbation is trivial and \emph{a gravitational BPR wave of arbitrary strength is completely transparent to the transverse perturbation wave}.

%

\section{Transparency of Strong Gravitational Waves} \label{sec:Transparency}
As was mentioned earlier, asymptotically $\alpha_0$ is always a linear function of $\eta$. In this section, we will study the case where $\alpha_0$ approaches a constant asymptotically. Without loss of generality, we will assume that $|\alpha_0(\eta)| \rightarrow 1$ as $|\eta|\rightarrow \infty$.

By redefining the constant of integration $K(\Omega,k)$ in Eq. (\ref{eq:PsiReflection}), we may rewrite the 'wave-function' as
\begin{equation} \label{eq:PsiWellBehaved}
\Psi = \frac{K(\Omega,k)}{|\alpha_0| ^{3/2}} \exp \left[- \frac{ik^2}{\Omega}\eta + \frac{ik^2}{\Omega} \int_\eta ^\infty \left(\frac{1}{\alpha_0 ^2} - 1\right) d\eta \right].
\end{equation} 
This means after the collision the asymptotic behavior of $\Psi$ is
\begin{equation} \label{eq:PsiInfty1}
\Psi = K(\Omega,k) \exp \left[- \frac{ik^2}{\Omega}\eta\right] \text{ as } \eta \rightarrow \infty
\end{equation}
and before the collision its asymptotic behavior is
\begin{equation} \label{eq:PsiMInfty1}
\Psi = K(\Omega,k) \exp \left[- \frac{ik^2}{\Omega}(\eta-\Delta)\right] \text{ as } \eta \rightarrow -\infty
\end{equation}
where 
\begin{equation} \label{eq:Delta}
\Delta = \int _{-\infty} ^\infty \left(\frac{1}{\alpha_0 ^2} - 1\right) d\eta
\end{equation}
If we set
\begin{equation}
K(\Omega,k) = K(\Omega) \delta(k+s \Omega)
\end{equation}
for $0<s<1$, then Eq. (\ref{eq:PsiInfty1}) gives the form of the wave after the collision,
\begin{equation} \label{eq:PsiInfty2}
\Psi \sim K(\Omega) \delta(k+s \Omega) \exp \left[- i s^2 \Omega \eta\right] \text{ as } \eta \rightarrow \infty.
\end{equation}
Taking the inverse Fourier transform gives the explicit form of the transmitted wave
\begin{equation} \label{eq:PsiInfty3}
\Psi \sim \frac{1}{2\pi} \int K(\Omega) e^{i\Omega\left[\zeta - s^2 \eta - s x^2\right]} d\Omega \text{ as } \eta \rightarrow \infty.
\end{equation}
This means that after the collision, $\psi$ is a wave of the form
\begin{equation} \label{eq:PsiInfty4}
\Psi = \Psi(\phi)  \text{ as } \eta \rightarrow \infty,
\end{equation}
where
\begin{equation} \label{eq:Phi}
\phi=\zeta - s^2 \eta - s x^2.
\end{equation}
Let us assume that the wave $\Psi$ is maximal (i.e. that its modulus $|\Psi|$ is maximal) at $\phi=0$. Restoring back the original coordinates $x^0$ and $x^1$ via Eq. (\ref{eq:EtaZeta}) and setting $x^0=0$ shows that the position of the transmitted wave satisfies
\begin{equation}
x^1 = \frac{2s}{1-s^2} x^2
\end{equation}
giving a transmitted wave with a scattering angle $\theta$, satisfying
\begin{equation}
\tan \theta = \frac{2s}{1-s^2}.
\end{equation}
In order to understand the behavior of the wave prior to the collision, we proceed in a similar fashion and take the inverse Fourier transform of Eq. (\ref{eq:PsiInfty2}). An analogous computation yields
\begin{equation} \label{eq:PsiMInfty4}
\Psi = \Psi(\phi + s^2 \Delta)  \text{ as } \eta \rightarrow -\infty,
\end{equation}
where $\Delta$ was defined in Eq. (\ref{eq:Delta}). Therefore the incident wave has the same shape and direction as the transmitted wave. In other words, a small transmitted wave would go through a strong BPR wave \emph{without any reflection} (see figure \ref{fig:BPRTransparency}). Its direction is preserved after the collision, and \emph{the only remnant of the collision is a phase shift} $s^2 \Delta$ as given by Eq. (\ref{eq:PsiMInfty4}). This is one of the defining properties of solitons \cite{bib:TheoryOfSolitons}.

Let us consider a practical example for the stability of the Bondi-Pirani-Robinson wave. Let $\alpha_0=\tanh \eta$ and $\Lambda_0'=\frac{\pm 1}{\cosh \eta}$ be the BPR wave that was mentioned in Eq. (\ref{eq:BargmannSolution}). For this wave the integral in Eq. (\ref{eq:PsiReflection}) can be evaluated explicitly,
\begin{equation}
\Psi(\eta,\Omega,k) = \frac{K}{|\tanh \eta| ^{3/2}} \exp \left[- \frac{ik^2}{\Omega} (\eta - \coth\eta) \right]
\end{equation}
Consider now a collision between the incident wave $\Psi$ (for fixed frequencies $k$ and $\Omega$) and a wave barrier. Normally, such a collision generates a reflected wave. However, here the asymptotic behavior of $\Psi$ is
\begin{equation} \label{eq:PsiAsymptotics}
\Psi \sim K \exp \left[- \frac{ik^2}{\Omega} (\eta \mp 1) \right] \text{ as } \eta \rightarrow \pm \infty
\end{equation}
showing that the wave $\Psi$ maintains the same asymptotic behavior except for a phase shift. Therefore the wave $\Psi$ is transparent to a BPR wave of arbitrary amplitude (see figure \ref{fig:BPRTransparency}). This is a \emph{striking fact and a strong hint of integrability}.

\begin{figure}
  \begin{center}
    \includegraphics[width=0.45\textwidth]{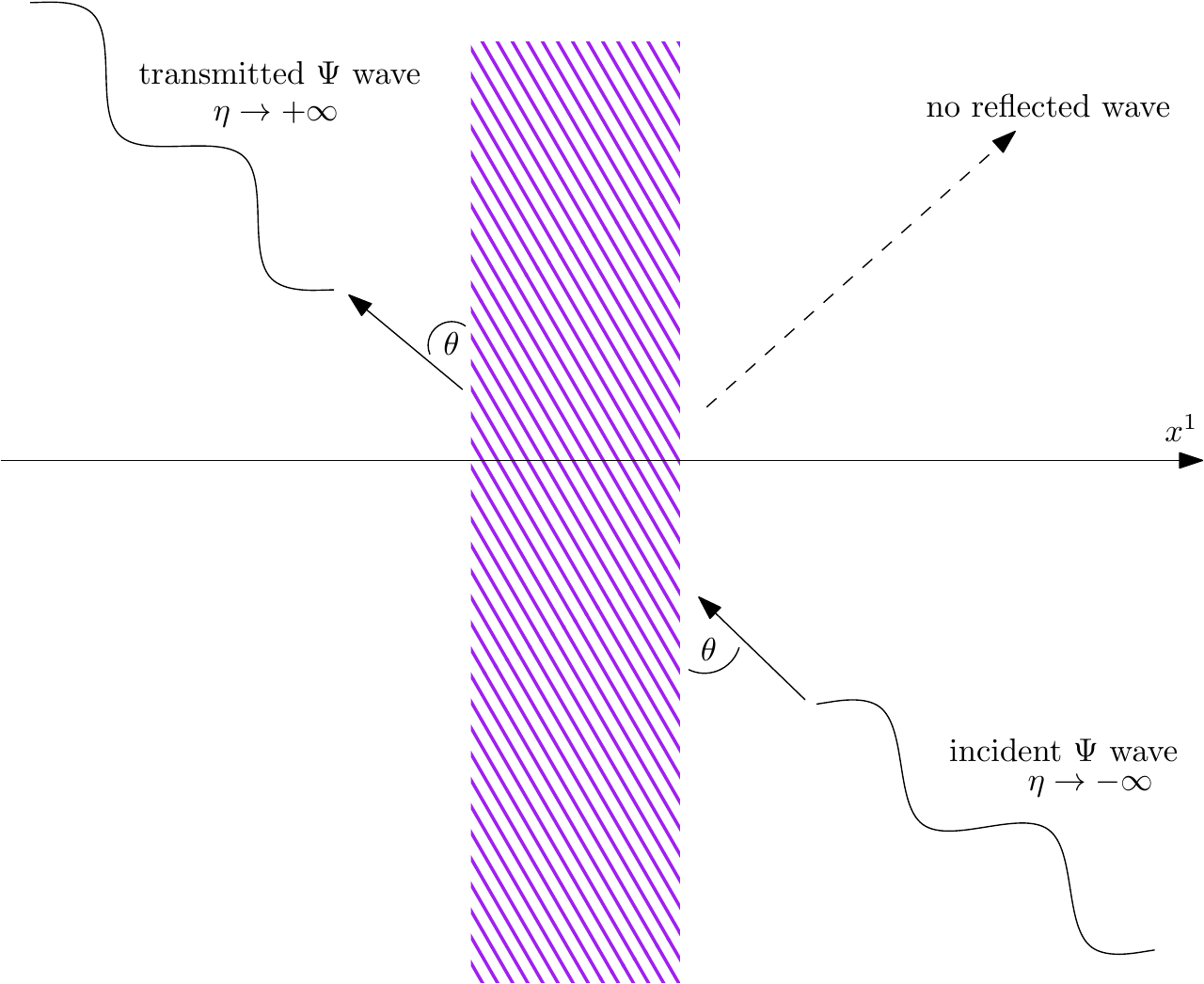}
  \end{center}
  \caption[Reflectionless Bondi-Pirani-Robinson wave]{The diagram shows the `wave-function' $\Psi$ for corresponding to $(\Lambda_0 ')^2$ being a Bargmann potential (\ref{eq:BargmannPotential}). In this case, Eq. (\ref{eq:PsiAsymptotics}) shows that $\Psi$ represents a wave traveling to the right both before and after hitting the wave barrier. Therefore there is no reflective wave and the barrier is transparent for $\Psi$. The only remnant of the interaction is a phase shift. \label{fig:BPRTransparency}}
\end{figure}

\section{Conclusions} \label{sec:Conclusions}

In studying the stability of the Bondi-Pirani-Robinson wave and that of the Schwarzschild black hole, we noticed a remarkable analogy between the problems. It is therefore very reasonable to expect that stability follows in more general circumstances. The fact that both the Bondi-Pirani-Robinson wave and the Schwarzschild black hole are a special case of the block diagonal metric integrable by the inverse scattering transform \cite{bib:BelinskiZakharov1}, makes it tempting to conjecture that perhaps all the solutions of the Einstein vacuum equation belonging to this class of metrics are stable. 

In both cases $k=0$ and $k\neq 0$ imaginary modes yield a spacetime that is not asymptotically flat and can therefore be disregarded on physical grounds. Consequently, we conclude that there are no unstable solutions for the perturbation, and that \emph{the Bondi-Pirani-Robinson (BPR) wave is stable}.

Beyond the mere stability, the case where $\alpha_0 \xrightarrow {\eta \rightarrow \pm 1} \pm1$ (where $(\Lambda_0 ')^2$ is a Bargmann potential (\ref{eq:BargmannPotential})) was studied further. In this case, an emitted wave traveling for a head-on collision towards a BPR wave of arbitrary amplitude revealed a fascinating physical phenomena. The emitted wave travelled through the BPR wave \emph{with no reflection}. Moreover, \emph{the transmitted wave left the collision process intact}, with its original shape and direction. The only hint of the collision process was a phase shift (see figure \ref{fig:BPRTransparency}). This phenomenon is common to integrable systems that describe solitons \cite{bib:TheoryOfSolitons}, and makes one suspect that the $1+2$ vacuum Einstein equations for diagonal metrics (\ref{eqs:DiagonalEinstein3DOffDiagonal},\ref{eqs:DiagonalEinstein3DDiagonal}) are integrable. 

\section{Acknowledgements}
Y. Hadad would like to thank Prof. Kundt for a fascinating correspondence that helped in forming the results of this work.

\end{document}